\documentclass[10pt]{llncs}
\usepackage[dvips]{graphicx}
\usepackage{latexsym,tabularx}
\usepackage{amsmath,amssymb}
\usepackage{url}

\newtheorem{fact}{{\em Fact}}

\begin{document}

\title{
A Succinct Grammar Compression\thanks{
This study was supported by KAKENHI(23680016,24700140) and JST PRESTO program.
}
}

\author{
Yasuo Tabei\inst{1}, 
Yoshimasa Takabatake\inst{2}, and
Hiroshi Sakamoto\inst{2,3}
}
\institute{
ERATO Minato Project, JST, Japan
\and
Kyushu Institute of Technology, Japan
\and PRESTO, JST, Japan\\
\email{
tabei.y.aa@m.titech.ac.jp,
\{takabatake, hiroshi\}@donald.ai.kyutech.ac.jp}
}

\maketitle

\begin{abstract}
We solve an open problem related to an optimal encoding of 
a {\em straight line program} (SLP), a canonical form of grammar compression
deriving a single string deterministically.
We show that an information-theoretic lower bound for representing an SLP with $n$ symbols 
requires at least $2n+\log n!+o(n)$ bits. 
We then present a succinct representation of an SLP; this representation is asymptotically equivalent to the lower bound.
The space is at most $2n\log\rho(1+o(1))$ bits for $\rho \leq 2\sqrt{n}$, while supporting random access to any production rule of 
an SLP in $O(\log\log n)$ time.
In addition, we present a novel dynamic data structure associating a digram with a unique symbol.
Such a data structure is called a {\em naming function} and has been implemented using a hash table that has a space-time tradeoff.
Thus, the memory space is mainly occupied by the hash table during the development of production rules.
Alternatively, we build a dynamic data structure for the naming function by leveraging the idea behind the {\em wavelet tree}.
The space is strictly bounded by $2n\log n(1+o(1))$ bits, while supporting $O(\log n)$ query and update time.
\end{abstract}

\section{Introduction}
Grammar compression has been an active research area since at least the seventies.
The problem consists of two phases: 
(i) building the smallest\footnote{This is almost equal to minimizing the number of variables in $G$.} 
context-free grammar (CFG) generating an input string uniquely 
and (ii) encoding an obtained CFG as compactly as possible. 

The phase (i) is known as an NP-hard problem which can not be approximated 
within a constant factor~\cite{Lehman-Shelat02}.
Therefore, many researchers have made considerable efforts to design grammar compressions 
achieving better approximation results in the last decade.
Charikar et al.~\cite{Charikar05} and Rytter~\cite{Rytter03} independently 
proposed the first $O(\log\frac{u}{g})$-approximation algorithms based on balanced 
grammar construction for the length $u$ of a string and the size $g$ of the smallest CFG. 
Later, Sakamoto~\cite{Sakamoto05} also developed an $O(\log\frac{u}{g})$-approximation algorithm 
based on an idea called pairwise comparison. 
In particular, Lehman~\cite{Lehman-Shelat02} proved that LZ77~\cite{LZ77} 
achieved the best approximation of $O(\log{n})$ under the condition of an unlimited window size.
Since the minimum {\em addition chain} problem is a special case of the problem 
of finding the smallest CFG~\cite{LehmanPhD}, modifying the approximation algorithms proposed so far is a difficult problem. 
Thus, the problem of grammar compression is pressing in the phase (ii). 

A straight line program (SLP) is a canonical form of a CFG, 
and has been used in many grammar compression algorithms~\cite{LZ78,Larsson00,LZ77,Apostolico00,ESP}.
The production rules in SLPs are in Chomsky normal form 
where the right hand side of a production rule in CFGs is a {\em digram}: a pair of symbols.
Thus, if $n$ symbols are stored in an array called a {\em phrase dictionary} 
consisting of $2n$ fixed-length codes each of which is represented by $\log{n}$ bits, 
the memory of the dictionary is $2n \log{n}$ bits, resulting in the memory for storing an input string usually being exceeded.
Although directly addressable codes achieving entropy bounds on strings whose memory consumption is the same as that of the fixed-length codes in the worst case have been presented~\cite{ferragina2007simple,sadakane2006squeezing,gonzalez2006statistical}, 
there are no codes that achieve an information-theoretic lower bound of storing an SLP in a phrase dictionary.
Since a nontrivial information-theoretic lower bound of directly addressable codes for a phrase dictionary remains unknown, 
establishing the lower bound and developing novel codes for optimally representing an SLP are challenges.

We present an optimal and directly addressable SLP within a strictly bounded memory close to 
the amount of a plain representation of the phrase dictionary.
We first give an information-theoretic lower bound on the problem of encoding an SLP, 
which has been unknown thus far. 
Let $C$ be a class of objects.
Representing an object $c \in C$ requires at least $\log|C|$ bits.
A representation of $c$ is succinct if it requires at most $\log{|C|}(1+o(1))$ bits.
Considering the facts and the characteristics of SLPs indicated in \cite{ESP}, 
one can predict that the lower bound for the class of SLPs with $n$ symbols 
would be between $2n$ and $4n+\log{n!}$. 
By leveraging this prediction, we derive that a lower bound of bits to represent SLPs is $2n+\log{n!}$.

We then present an almost optimal encoding of SLPs based on
{\em monotonic subsequence decomposition} of a sequence.
Any permutation of $[1,n]$ is decomposable into at most $\rho \leq 2\sqrt{n}$
monotonic subsequences in $O(n^{1.5})$ time~\cite{Yehuda1998} and
there is a $1.71$-approximation\footnote{Minimizing $\rho$ is NP-hard.} 
algorithm in $O(n^3)$ time~\cite{Fomin2002}. 
While the previous encoding method for SLPs presented in~\cite{Takabatake2012} 
is also based on the decomposition, the size is not asymptotically equal
to the lower bound when $\rho \simeq \sqrt{n}$.
We improve the data structure by using the {\em wavelet tree} (WT)~\cite{Grossi03} and its improved results~\cite{Jeremy2010,Golynski2006} such that
our novel data structure achieves the smaller bound of $\min\{2n+n\log n +o(1),2n\log\rho(1+o(1))\}$ 
bits for any SLP with $n$ symbols while supporting $O(\log\log\rho)$ access time.
Our method is applicable to any types of algorithm generating SLPs
including Re-Pair~\cite{Larsson00} and an online algorithm called LCA~\cite{Maruyama2012}.
Barbay et al.~\cite{Barbay2009} presented a succinct representation of a sequence using the monotonic subsequence decomposition. 
Their method uses the representation of an integer function built on a succinct representation of integer ranges. 
Its size is estimated to be the {\em degree entropy} of an ordered tree~\cite{Jansson2012-JCSS}.

Another contribution of this paper is to present a dynamic data structure for checking whether or not a production rule in a CFG has been generated in execution. 
Such a data structure is called a {\em naming function}, and is also necessary for practical grammar compressions. 
When the set of symbols is static, we can construct a perfect hash as a naming function in linear time, which achieves 
an amount of space within around a factor of 2 from the information-theoretical minimum~\cite{Botelho2007}.
However, variables of SLPs are generated step by step in grammar compression.
While the function can be dynamically computed by a randomization~\cite{Karp87} or
a deterministic solution~\cite{Karp72} in $O(1)$ time and linear space,
a hidden constant in the required space was not clear.
We present a dynamic data structure to compute function values in $O(\log n)$ query time
and update time. The space is strictly bounded by $2n\log n(1+o(1))$ bits.

\section{Preliminaries}

\subsection{Grammar compression}
For a finite set $C$, $|C|$ denotes its cardinality.
{\em Alphabet} $\Sigma$ is a finite set of letters and $\sigma=|\Sigma|$ is a constant. 
$\cal X$ is a recursively enumerable set 
of {\em variables} with $\Sigma\cap \cal{X}=\emptyset$.
A sequence of symbols from $\Sigma\cup \cal{X}$ is called a string.
The set of all possible strings from $\Sigma$ is denoted by $\Sigma^*$.
For a string $S$, the expressions $|S|$, $S[i]$, and $S[i,j]$ 
denote the length of $S$, the $i$-th symbol of $S$,
and the substring of $S$ from $S[i]$ to $S[j]$, respectively.
Let $[S]$ be the set of symbols composing $S$.
A string of length two is called a {\em digram}.

A CFG is represented by ${\cal G}=(\Sigma,V,P,X_s)$ 
where $V$ is a finite subset of $\cal X$, $P$ is a finite subset of 
$V\times (V\cup\cal{X})^*$, and $X_s\in V$.
A member of $P$ is called a production rule and $X_s$ is called the start symbol.
The set of strings in $\Sigma^*$ derived from $X_s$ by ${\cal G}$ is 
denoted by $L({\cal G})$.

A CFG ${\cal G}$ is called {\em admissible} if exactly one $X\to \alpha\in P$ exists and $|L({\cal G})|=1$.
An admissible ${\cal G}$ deriving $S$ is called a grammar compression of $S$ for any $X\in V$.

We consider only the case $|\alpha|=2$ for any production rule $X\to \alpha$
because any grammar compression with $n$ variables can be transformed into 
such a restricted CFG with at most $2n$ variables.
Moreover, this restriction is useful for practical applications of compression algorithms, e.g.,
LZ78~\cite{LZ78}, REPAIR~\cite{Larsson00}, and LCA~\cite{Maruyama2012},
and indices, e.g., SLP~\cite{Claude09} and ESP~\cite{Maruyama2011}.

The derivation tree of $G$ is represented by a rooted ordered binary tree
such that internal nodes are labeled by variables in $V$ and
the {\em yields}, i.e., the sequence of labels of leaves is equal to $S$.
In this tree, any internal node $Z\in V$ has 
a left child labeled $X$ and a right child labeled $Y$, which 
corresponds to the production rule $Z\to XY$.

If a CFG is obtained from any other CFG by a permutation 
$\pi:\Sigma\cup V\to \Sigma\cup V$, they are identical to each other
because the string derived from one is transformed to 
that from the other by the renaming.
For example, $P=\{Z\to XY,Y\to ab,X\to aa\}$ and
$P'=\{X\to YZ,Z\to ab,Y\to aa\}$ are identical each other.
On the other hand, they are clearly different from 
$P''=\{Z\to aY,Y\to bX,X\to aa\}$ because their depths are different.
Thus, we assume the following canonical form of CFG called 
{\em straight line program} (SLP).

\begin{definition}\label{SLP}(Karpinsk-Rytter-Shinohara~\cite{SLP})
An SLP is a grammar compression over $\Sigma\cup V$ 
whose production rules are formed by either 
$X_i\to a$ or $X_k\to X_iX_j$, where $a\in\Sigma$ and $1\leq i,j < k\leq |V|$.
\end{definition}

\subsection{Phrase/reverse dictionary}
For a set $P$ of production rules,
a {\em phrase dictionary} $D$ is a data structure for directly accessing
 the phrase $X_iX_j$ for any $X_k\in V$ if $X_k\to X_iX_j \in P$.
Regarding a triple $(k,i,j)$ of positive integers as $X_k\to X_iX_j$, 
we can store the phrase dictionary consisting of $n$ variables in an integer array $D[1,2n]$, 
where $D[2k-1]=D[2k]=0$ if $k$ belongs to an alphabet i.e., $1 \leq k \leq |\Sigma|$. 
$X_i$ and $X_j$ are accessible as $D[2k-1]$ and $D[2k]$ by indices $2k-1$ and $2k$ for $X_k$, 
respectively. 
A plain representation of $D$ using fixed-length codes requires $2n\log n$ bits of space to store 
$n$ production rules.

Reverse dictionary $D^{-1}$ is a data structure for directly accessing the variable $X_k$ 
given $X_iX_j$ for a production rule $X_k\to X_iX_j \in P$. 
Thus, $D^{-1}(X_iX_j)$ returns $X_k$ if $X_k \rightarrow X_iX_j \in P$.
A hash table is a representative data structure for $D^{-1}$ enabling $O(1)$ time access 
and achieving $O(n\log n)$ bits of space. 

\subsection{Rank/select dictionary}
We present a phrase dictionary based on the {\em rank/select dictionary}, 
a data structure for a bit string $B$~\cite{Jacobson89} supporting the following queries:
$\mbox{rank}_c(B,i)$ returns the number of occurrences of $c \in \{0,1\}$ in $B[1,i]$ and 
$\mbox{select}_c(B,i)$ returns the position of the $i$-th occurrence of $c \in \{0, 1\}$ in $B$.
For example, if $B=10110100111$ is given, then $\mbox{rank}_1(S,7)=4$ because the number of $1$s in 
$B[1,7]$ is $4$,
and $\mbox{select}_1(S,5)=9$ because the position of the fifth $1$ in $B$ is $9$.
Although naive approaches require the $O(|B|)$ time to compute a rank, 
several data structures with only the $|B|+o(|B|)$ bit storage to 
achieve $O(1)$ time~\cite{Navarro12,Okanohara07} have been presented. 
Most methods compute a select query by a binary search on a bit string $B$ in $O(\log |B|)$ time. 
A data structure for computing the select query in $O(1)$ time has also been presented~\cite{Ram02}.

\subsection{Wavelet tree}
A WT is a data structure for a string $S \in \Sigma^*$, and 
it can be used to compute the rank and select queries on a string $S$ over an ordinal alphabet 
in $O(\log{\sigma})$ time and $n\log\sigma(1+o(1))$ bits~\cite{Grossi03}.
Data structures supporting the rank and select queries in $O(\log\log{\sigma})$ time 
with the same space have been proposed~\cite{Golynski2006,Jeremy2010}. 
WT also supports $\mbox{access}(S,i)$ which returns $S[i]$ in $O(\log{\sigma})$ time.
Recently, WT has been extended to support various operations on strings~\cite{Navarro2012-cpm}.

A WT for a sequence $S$ over $\Sigma=\{1,...,\sigma\}$ is a binary tree that 
can be, recursively, presented over a sub-alphabet range $[a,b] \subseteq [1,\sigma]$. 
Let $S_v$ be a sequence represented in a node $v$, and let 
$left(v)$ and $right(v)$ be left and right children of node $v$, respectively. 
The root $v_{root}$ represents $S_{root}=S$ over the alphabet range $[1,\sigma]$.
At each node $v$, $S_{v}$ is split into two subsequences $S_{left(v)}$ consisting of 
the sub-alphabet range $[a,\lfloor \frac{(a+b)}{2} \rfloor]$ for $left(v)$ 
and $S_{right(v)}$ consisting of the sub-alphabet range $[\lfloor \frac{(a+b)}{2} \rfloor + 1, b]$ 
for $right(v)$ where $S_{left(v)}$ and $S_{right(v)}$ keep the order of elements in $S_v$. 
The splitting process repeats until $a=b$. 
Each node $v$ in the binary tree contains a rank/select dictionary on a bit string $B_v$. 
Bit $B_v[k]$ indicates whether $S_v[k]$ should be moved to $left(v)$ or $right(v)$. 
If $B_v[k]=0$, $S_{left(v)}$ contains $S_{v}[k]$. 
If $B_v[k]=1$, $S_{right(v)}$ inherits $S_v[k]$. 
Formally, $B_v[k]$ with an alphabet range $[a,b]$ is defined as:
\[
B_v[k] = \left\{ \begin{array}{ll}
1 & \mbox{if} \ S_v[k] > \lfloor (a+b)/2 \rfloor \\ 0 & \mbox{if} \ S_v[k] \le \lfloor
(a+b)/2 \rfloor \end{array}
\right. .
\]

An example of a WT is shown in Figure~\ref{fig:wavelet_tree_ex}.
In this example, since $S_{root}[2]=4$ belongs to the higher half $[3,4]$ of an alphabet range $[1,4]$ represented in the root; therefore, it is the second element of $S_{root}$ that must go to the right child of the root, $B_{root}[2]=1$ and $S_{right(root)}[2]=S_{root}[2]=4$.

\begin{figure*}[t]
\begin{center}
\includegraphics[width=0.4\textwidth]{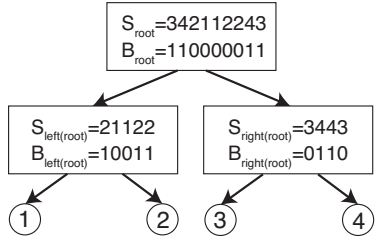}
\end{center}
\vspace{-0.5cm}
\caption{Example of wavelet tree for a sequence $S=342112243$ over an alphabet $\{1,2,3,4\}$.}
\label{fig:wavelet_tree_ex}
\end{figure*}

\section{Succinct SLP}

\subsection{Information-theoretic lower bound}
In this section, we present a tight lower bound to represent SLPs
having a set of production rules $P$ consisting of $n=|\Sigma \cup V|$ symbols.
Each production rule $Z\to XY \in P$ is considered as two directed edges $(Z,X)$ and $(Z,Y)$, 
the SLP can be seen as a directed acyclic graph (DAG) with a single source and $|\Sigma|$ sinks.
Here, we consider $(Z,X)$ as the left edge and $(Z,Y)$ as the right edge.
In addition, $P$ can be considered as a DAG with the single source and 
with a single sink 
by introducing a super-sink $s$ and drawing directed left and right edges from any sink to $s$~(Figure~\ref{fig:dag}). 
Let ${\cal DAG}(n)$ be the set of all possible $G$s with $n$ nodes and 
${\cal DAG}=\bigcup_{n\to\infty}{\cal DAG}(n)$. 
Since two SLPs are identical if an SLP can be converted to the other SLP by 
a permutation $\pi : \Sigma \cup V \rightarrow \Sigma \cup V$,
the number of different SLPs is $|{\cal DAG}(n)|$.
Any internal node of $G \in {\cal DAG}(n)$ has exactly two (left/right) edges.
Thus, the following fact remarked in~\cite{Maruyama2011} is true.

\begin{figure*}[t]
\begin{center}
\includegraphics[width=0.9\textwidth]{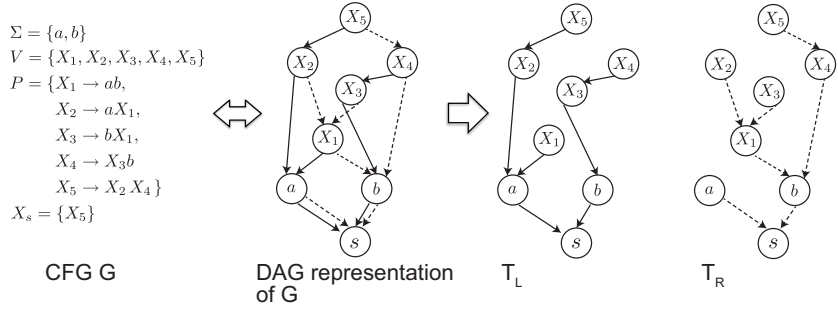}
\end{center}
\vspace{-0.7cm}
\caption{Example of DAG representation of an SLP and its spanning tree decomposition. 
An SLP is represented by a DAG $G$. $G$ is decomposed into the left tree $T_L$ and right tree $T_R$.}
\label{fig:dag}
\end{figure*}

\begin{fact}
An in-branching spanning tree is an ordered tree such that the out-degree of any node 
except the root is exactly one. 
For any in-branching spanning tree of $G$, the graph consisting of the
remaining edges and their adjacent vertices is also an in-branching spanning tree of $G$.
\end{fact}

The in-branching spanning tree consisting of the left edges (respectively the right edges) and their adjacent vertices is called the
{\em left tree} $T_L$ (respectively {\em right tree} $T_R$) of $G$.
Note that the source in $G$ is a leaf of both $T_L$ and $T_R$, and
the super-sink of $G$ is the root of both $T_L$ and $T_R$.
We shall call the operation of decomposing a DAG $G$ into two spanning trees $T_L$ and $T_R$ 
{\em spanning tree decomposition}. 
In Figure~\ref{fig:dag}, the source $x_5$ in $G$ is a leaf of both $T_L$ and $T_R$, and 
the super-sink $s$ in $G$ is the root of both $T_L$ and $T_R$.

Any ordered tree is an elements in ${\cal T}=\bigcup_{n\to\infty}{\cal T}_n$
where ${\cal T}_n$ is the set of all possible ordered trees with $n$ nodes.
As shown in~\cite{Asai2002,Zaki2002}, 
there exists an enumeration tree for ${\cal T}$ such that
any $T\in {\cal T}$ appears exactly once. 
The enumeration tree is defined by the {\em rightmost expansion}, i.e.,
in this enumeration tree, a node $T'\in {\cal T}_{n+1}$, 
which is a child of $T\in {\cal T}_n$, is obtained by adding a rightmost node to $T$.
In our problem, an ordered tree $T\in {\cal T}_{n+1}$ is 
identical to a left tree $T_L$ with $n+1$ nodes for $n=|\Sigma\cup V|$ symbols.

Let $G\oplus (u,v)$ be the DAG obtained by adding the edge $(u,v)$ to a DAG $G$.
If necessary, we write $G\oplus (u,v)_L$ to indicate that
$(u,v)$ is added as a left edge.
For a set $E$ of edges, the DAG $G\oplus E$ is defined analogously.
The DAG $G\oplus E$ is defined as adding all the edges $(u,v) \in E$ to $G$.
The DAG $G\ominus E$ is also defined as deleting all the edges $(u,v) \in E$ from $G$.

\begin{theorem}\label{th1}
The information-theoretic lower bound on the minimum number of bits needed to
represent an SLP with $n$ symbols is $2n + \log n! + o(n)$.
\end{theorem}
\begin{proof}
Let ${\cal S}(n)$ be the set of all possible DAGs with $n$ nodes
and a single source/sink such that any internal node has exactly two children.
This ${\cal S}(n)$ is a super set of ${\cal DAG}(n)$ because
the in-degree of the sink of any DAG in ${\cal DAG}(n)$ must be
exactly $2\sigma$, whereas ${\cal S}(n)$ does not have such a restriction.
By the definition, 
$|{\cal S}(n)|/n^\sigma \leq 
|{\cal DAG}(n)| \leq 
|{\cal S}(n)|$ holds.

Let ${\cal S}(n,T)=\{G\in {\cal S}(n)\mid G=T\oplus T_R,\;T_R\in {\cal T}_n\}$.
We show $|{\cal S}(n,T)|=(n-1)!$ for each $T\in {\cal T}_n$
by induction on $n\geq 1$.
Since the base case $n=1$ is clear, we assume that the induction hypothesis 
is true for some $n\geq 1$.

Let $T'_L$ be the rightmost expansion of $T_L$ such that
the rightmost node $u$ is added as the rightmost child of node $v$ in $T_L$,
and let $G'\in {\cal S}(n+1,T'_L)$ with a left tree $T'_L$. 
By the induction hypothesis,
the number of $G\in {\cal S}(n,T_L)$ is $(n-1)!$
and $T_L$ is embedded into $G$ as the left tree.
Then, $G'$ is constructed by adding the left edge $(u,v)$
and a right edge $(u,x)$ for a node $x$ in $T_L$.

Let $s$ be the source of $G$.
For $v=s$, each $G'=G\oplus (u,v)_L\oplus (u,x)_R\in {\cal S}(n+1,T'_L)$ is admissible,
and the number of them is clearly $n|{\cal S}(n,T_L)|=n!$.
For $v\neq s$, if $x=s$, 
$G'=G\oplus (u,v)_L\oplus (u,x)_R\in {\cal S}(n+1,T'_L)$ is admissible.

Otherwise, there exists the lowest common ancestor $y$ of $s$ and $x$
on $T_R$ with $G=T_L\oplus T_R$.
Let $z$ be the unique child of $y$ and let $p(z',z)$ be the path of $T_R$ 
from $z'$ to $z$, where possibly $s=z$.
If the in-dgree of any node in $p(s,z)$ is at most one in $G$, 
we generate $G' = G\oplus (u,v)_L \oplus (u,x)_R \ominus (z,y)_R \oplus (z,u)_R$.
If $G'$ contains a cycle, it must contain the edge $(z,u)$.
However, there is no such a path because of the condition of $p(s,z)$.
Thus, $G'$ is an admissible DAG in ${\cal S}(n+1,T'_L)$.
Conversely, if some node in $p(s,z)$ is more than two in $G$, let $z'$ be the nearest one from $s$.
Analogously, $G' = G\oplus (u,v)_L \oplus (u,x)_R \ominus (z',z)_R \oplus (z',u)_R$
is an admissible DAG in ${\cal S}(n+1,T'_L)$.

In all the cases, the number of such $G'$s is also $n!$ because 
no edge is changed in $T_L$ and the pair $(T'_L,T'_R)$ 
containing the edge $(u,x)_R$ is unique for any fixed $T'_L$.
Thus, $|{\cal S}(n+1,T)|=n!$ is true for each $T\in {\cal T}_{n+1}$.

This result derives $|{\cal S}(n)|=C_n(n-1)!$
where $C_n=\frac{1}{n+1}{2n \choose n}\simeq 2^{2n}n^{-3/2}$
is the number of ordered trees with $n+1$ nodes.
Combining this with 
$|{\cal S}(n)|/n^\sigma \leq 
|{\cal G}(n)| \leq 
|{\cal S}(n)|$ as well,
we get the result that the information-theoretic minimum bits needed to represent
$G\in {\cal DAG}(n)$ is at least $2n + \log n! + o(n)$.
\hspace{\fill}$\Box$
\end{proof}

\subsection{An optimal SLP representation}
\begin{figure}[t]
\begin{center}
\includegraphics[width=0.7\textwidth]{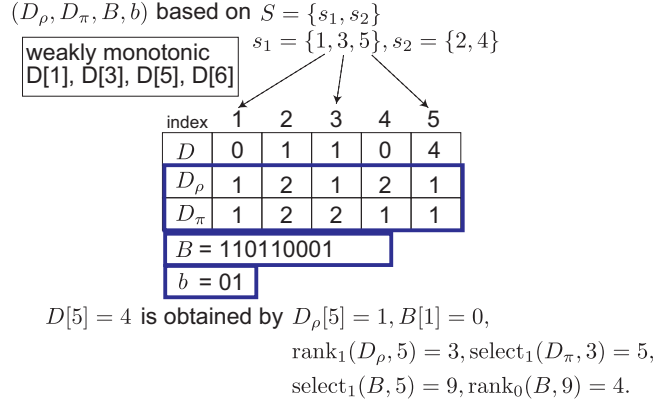}
\end{center}
\vspace{-.5cm}
\caption{
{\bf Encoded phrase dictionary:}
$D$ indicates the remaining sequence $D[2],D[4],\ldots, D[2n]$.
$D$ is encoded by $(D_\rho,D_\pi,{\bf B},{\bf b})$
based on a monotonic decomposition ${\cal S}$ of $D$, i.e.,
each $s\in{\cal S}$ indicates a weakly monotonic subsequence in $D$;
$D_\rho$ is the sequence of $i$ indicating the membership for some $s_i\in{\cal S}$,
$D_\pi$ is a permutation of $D_\rho$ with respect to 
the corresponding value in $D$,
${\bf B}$ is a binary encoding of the sorted $D$ in increasing order.
We show only the case that $D[i]$ is a member of an increasing $s\in{\cal S}$, but 
the other case is similarly computed by ${\bf b}$.
}
\label{dic_enc}
\end{figure}

We present an optimal reresentation of an SLP as an improvement of the data structure recently presented in \cite{Takabatake2012}.
We apply the spanning tree decomposition to the DAG $G$ of a given SLP, and obtain the DAG $T_L \oplus T_R(=G)$.
We rename the variables in $T_L$ by breadth-first order 
and also rename variables in $T_R$ according to the $T_L$.
Let $G^\prime$ be the resulting DAG from $G$. 
Then, for the array representation $D[1,2n]$ of $G'$, we obtain
the condition $D[1]\leq D[3]\leq \ldots \leq D[2n-1]$.
Since this monotonic sequence is encoded by $2n + o(n)$ bits,
$D$ is represented by $2n + n\log n + o(n)$ bits supporting
${\mbox access}(D,k)$ $(1\leq k\leq 2n)$ in $O(1)$ time.
We focus on the remaining sequence of length $n$, i.e., $D[2],D[4],\ldots,D[2n]$.
For simplicity, we write $D$ instead of $[D[2],D[4],\ldots,D[2n]]$. 

Let ${\cal S}=\{s_1,\ldots,s_\rho\}$ be a disjoint set of subsequences of $[1,n]$ such that 
any $i  \in \{1,2,...,n\}$ is contained in some $s_k$ and any $s_i,s_j$ $(i \neq j)$ are disjoint.
Such an ${\cal S}$ is called a decomposition of $D$.
A sequence $D[s_{k_1}],\ldots,D[s_{k_p}]$ is weakly monotonic 
if it is {\em increasing}, i.e., $D[s_{k_1}]\leq\ldots\leq D[s_{k_p}]$ or {\em decreasing}, i.e., $D[s_{k_1}]\geq\ldots\geq D[s_{k_p}]$.
In addition, ${\cal S}$ is called {\em monotonic} if the sequence $D[s_{k_1}],\ldots,D[s_{k_p}]$ is weakly monotonic 
for any $s_k=[s_{k_1},\ldots,s_{k_p}]\in {\cal S}$.

\begin{theorem}\label{th2}
Any SLP with $n$ symbols can be represented using $2n\log\rho(1 + o(1))$ bits
for $\rho\leq 2\sqrt{n}$, while supporting $O(\log\log\rho)$ access time. 
\end{theorem}
\begin{proof}
It is sufficient to prove that any $D$ of length $n$ 
can be represented using $2n\log\rho + o(n)$ bits for some $\rho\leq 2\sqrt{n}$.
By the result in~\cite{Yehuda1998},
we can construct a monotonic decomposition ${\cal S}$ of $D$ such that
$\rho = |{\cal S}| \leq 2\sqrt{n}$.

We represent the sequence $D$ as a four-tuple $(D_\rho,D_\pi,{\bf B},{\bf b})$ using ${\cal S}$.
For each $1\leq p\leq  n$, $D_\rho [p]=k$ iff 
$p$ is a member of $s_k\in {\cal S}$ for some $1\leq k\leq \rho$.
Let $(D[1],D_\rho[1]),\ldots,(D[n],D_\rho[n])$ be the sequence of pairs $(D[p],D_{\rho}[p])$ $(1 \leq p \leq n)$.
We sort these pairs with respect to the keys $D[p]$ $(1 \leq p \leq n)$ and 
obtain the sorted sequence $(D[\ell_1],D_\rho[\ell_1]),\ldots,(D[\ell_n],D_\rho[\ell_n])$. 
We define $D_\pi$ as the permutation $D_\rho[\ell_1]\cdots D_\rho[\ell_n]$. 

${\bf B}\in\{0,1\}^*$ is defined as the bit string
\[{\bf B}
=0^{D[\ell_1]}10^{D[\ell_2]-D[\ell_1]}\cdots1
0^{D[\ell_n]-D[\ell_{n-1}]}1.
\]

Finally, ${\bf b}[k]=0$ if $s_k\in {\cal S}$ is increasing and
${\bf b}[k]=1$ otherwise for $1\leq k\leq \rho$.
$D$ and $D_{\rho}$ are represented by WTs, respectively, and ${\bf B}$ is a rank/select dictionary.

We recover $D[p]$ using $(D_\rho,D_\pi,{\bf B},{\bf b})$.
When $D_\rho[p]=k$ and ${\bf b}[k]=0$, i.e., $D[p]$ is included in the 
$k$-th monotonic subsequence $s_k \in {\cal S}$ that is increasing, 
we obtain 
\[D[p] =
\mbox{rank}_0({\bf B},
\mbox{select}_1({\bf B},\ell)) 
\] 
by $\ell = \mbox{select}_k(
D_\pi,\mbox{rank}_k(D_\rho, p))$.
When $D_\rho[p]=k$ and ${\bf b}[k]=1$, we can similarly obtain $D[p]$ 
replacing $\ell$ by $r = \mbox{select}_k(D_\pi, 
(
\mbox{rank}_k(D_\rho,n)+1 - \mbox{rank}_k(D_\rho,p)
)
)$.

The total size of the data structure formed by $(D_\rho,D_\pi , {\bf B}, {\bf b})$
is at most $2n\log\rho(1+o(1))$ bits. 
The rank/select/access operations of the WT for a static sequence over $\rho\leq 2\sqrt{n}$ symbols 
can be improved to achieve $O(\log\log\rho)$ time for each query~\cite{Jeremy2010,Golynski2006}.
\hspace{\fill}$\Box$
\end{proof}

In Figure~\ref{dic_enc}, for the sequence $(0,1),(1,2),(1,1),(0,2),(4,1)$ of pairs $(D[p],D_{\rho}[p])$ $(1\leq p \leq 5)$, 
the sorted sequence is $(0,1),(0,2),(1,2),(1,1),(4,1)$.
Thus, $D_\pi$ is $12211$. 
${\bf B}=0^010^{(0-0)}10^{(1-0)}10^{(1-1)}10^{(4-1)}1=110110001$.
$b[1]=0$ because $s_1$ is increasing, and $b[2]=1$ because $s_2$ is decreasing.

\section{Data Structure for Reverse Dictionary}

In this section, we present a data structure for simulating 
the naming function $H$ defined as follows.
For a phrase dictionary $D$ with $n$ symbols,
\[
H(X_iX_j)=\left\{
\begin{array}{ll}
D^{-1}(X_iX_j), & \;\; \mbox{if} \;D[k]=X_iX_j \;\mbox{ for some } 1\leq k\leq n, \\
X_{n+1}, & \;\; \mbox{otherwise.}
\end{array}
\right.
\]

For a sufficiently large $V$, 
we set a total order on $(\Sigma\cup V)^2=\{ XY \mid
X,Y\in \Sigma\cup V\}$, i.e., 
the lexicographical order of the $n^2$ digrams.
This order is represented by the range $[1,n^2]$.
Then, we recursively define WT $T_D$ for a phrase dictionary $D$ 
partitioning $[1,n^2]$.
On the root node, the initial range $[1,n^2]$ is partitioned into two parts: a left range
$L[1,\lfloor(1+n^2)\rfloor/2]$ and a right range $R[\lfloor(1+n^2)\rfloor/2+1,n^2]$.
The root is the bit string ${\bf B}$ such that
${\bf B}[i]=0$ if $D[i]\in L$ and ${\bf B}[i]=1$ if $D[i]\in R$.
By this, the sequence of digrams, $D$, is decomposed into two subsequences $D_L$ and $D_R$;
they are projected on the roots of the left and right subtrees, respectively.
Each sub-range is recursively partitioned and 
the subsequence of $D$ on a node is further decomposed
with respect to the partitioning on the node.
This process is repeated until the length of any sub-range is one.
Let ${\bf B}_i$ be the bit string assigned to the $i$-th node
of $T_D$ in the breadth-first traversal.
In Figure~\ref{wavelet_dic}, we show an example
of such a data structure for a phrase dictionary $D$.

\begin{figure}[tb]
\begin{center}
\includegraphics[width=0.8\textwidth]{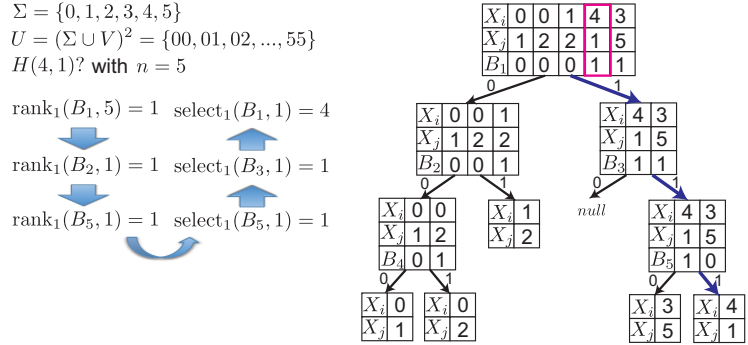}
\end{center}
\caption{
{\bf WT for reverse dictionary:}
The bit string ${\bf B}_i$ is assigned to 
the $i$-th node in breadth-first order.
For each internal node $i$,
we can move to the left child by $\mbox{rank}_0$
and to the right child by $\mbox{rank}_1$ on ${\bf B}_i$.
The upward traversal is simulated by $\mbox{select}_0$
and $\mbox{select}_1$ as shown.
The leaf for an existing digram is represented by $1$
and $null$ is represented by $0$, whereas
these bits are omitted in this figure.
}
\label{wavelet_dic}
\end{figure}

\begin{theorem}\label{wavelet-th}
The naming function for phrase dictionary $D$ over $n=|\Sigma\cup V|$ 
symbols can be computed by the proposed data structure $D_T$
in $O(\log n)$ time for any digram.
Moreover, when a digram does not exist in the current $D$,
$D_T$ can be updated in the same time and 
the space is at most $2n\log n(1+o(1))$ bits.
\end{theorem}
\begin{proof}
$D_T$ is regarded as a WT for 
a string $S$ of length $n$ such that any symbol is represented 
in $2\log n$ bits.
Thus, $H(XY)$ is obtained by $\mbox{select}_{XY}(S,1)$.
The query time is bounded by the number of
$\mbox{rank}$ and $\mbox{select}$ operations for bit strings
performed until the operation flow returns to the root.
Since the total range is $[1,n^2]$, i.e., the height of $T_D$ is at most $2\log n$,
the query time and the size are derived.
When $XY$ does not exist in $D$,
let $i_1,i_2,\ldots,i_k$ be the sequence of traversed nodes from the root $i_1$ to
a leaf $i_k$ and let ${\bf B}_{i_j}$ be the bit string on $i_j$. 
Given an access/rank/select dictionary for ${\bf B}_{i_j}$,
we can update it for ${\bf B}_{i_j}b$ and $b\in\{0,1\}$ in $O(1)$ time.
Therefore, the update time of $T_D$ for any digram is $O(k)=O(\log n)$.
\hspace{\fill}$\Box$
\end{proof}

\section{Discussion}

We have investigated three problems related to the construction of an SLP:
the information-theoretic lower bound for representing the phrase dictionary $D$,
an optimal representation of a directly addressable $D$,
and a dynamic data structure for $D^{-1}$.
Here, we consider the results of this study from the viewpoint of open questions.

For the first problem, we approximately estimated 
the size of a set of SLPs with $n$ symbols, which is
almost equal to the exact set.
This problem, however, has several variants, e.g.,
the set of SLPs with $n$ symbols deriving the same string,
which is quite difficult to estimate owing to the NP-hardness of the smallest CFG problem.
There is another variant obtained by a restriction:
Any two different variables do not derive the same digram, i.e.,
$Z\to XY$ and $Z'\to XY$ do not exist simultaneously for $Z\neq Z'$.
Although such variables are not prohibited in the definition of SLP,
they should be removed for space efficiency.
On the other hand, even if we assume this restriction,
the information-theoretic lower bound is never smaller than $\log n!$ bits because, 
given a directed chain of length $n$ as $T_L$,
we can easily construct $(n-1)!$ admissible DAGs.

For the second problem, we proposed almost optimal encoding of SLPs.
From the standpoint of massive data compression,
one drawback of the proposed encoding is that
the whole phrase dictionary must be stored in memory beforehand.
Since symbols must be sorted, we need a dynamic data structure
to allow the insertion of symbols in an array, e.g.,~\cite{Jansson2012}.
Such data structures, however, require $O(n\log n)$ bits of space.

For the last problem, the query time and update time of proposed 
data structure are both $O(\log n)$.
This cost is considerable and it is difficult to improve it
to $O(\log\log n)$ because $D$ is not static.
When focusing on the characteristics of SLPs, 
we can improve the query time probabilistically;
since any symbol $X$ appears in $D$ at least once and $|D|=2n$,
the average of frequency of $X$ is at most two.
Thus, using an additional array of size $n\log n$ bits,
we can check $H(XY)$ in $O(1)$ time with probability at least $1/2$.
However, improving this probability is not easy.
For this problem, achieving $O(1)$ amortized query time is also an interesting challenge.

\bibliographystyle{plain}

{\small
\bibliography{main.bib}
}
\end{document}